\documentclass[twocolumn]{autart}
\usepackage[utf8]{inputenc}
\DeclareUnicodeCharacter{221E}{$\infty$}
\usepackage{caption}
\usepackage{comment}
\usepackage{subcaption}
\usepackage{graphicx}
\usepackage{amsmath}
\usepackage{breqn}
\usepackage{enumitem}
\newlist{steps}{enumerate}{1}
\setlist[steps, 1]{label = Step \arabic*:}

\let\theoremstyle\relax
\usepackage{amsthm,amssymb}
\usepackage[version=4]{mhchem}
\usepackage{siunitx}
\usepackage{longtable,tabularx}
\usepackage{tikz}
\usepackage{pgfplots}
\pgfplotsset{compat=1.15}
\usepackage{standalone}
\usepackage[sort&compress,numbers]{natbib}

\newtheorem{lemma}{Lemma}
\newtheorem{remark}{Remark}
\newcommand{\R}{\mathbb{R}}

\newcommand{\tm}[1]{{\color{black}{#1}}}

\begin{document}
\begin{frontmatter}
\title{On the convexity of static output feedback control synthesis for systems with lossless nonlinearities}

\author[UMN]{Talha Mushtaq}\ead{musht002@umn.edu},    
\author[UMich]{Peter Seiler}\ead{pseiler@umich.edu},               
\author[UMN]{Maziar S. Hemati}\ead{mhemati@umn.edu}  

\address[UMN]{Aerospace Engineering and Mechanics,
University of Minnesota,
Minneapolis, MN 55455, USA}  
\address[UMich]{Electrical Engineering and Computer Science,
 University of Michigan, Ann Arbor, MI 48109, USA}             
\begin{abstract}
Computing a stabilizing static output-feedback (SOF) controller is an NP-hard problem, in general.
Yet, these controllers have amassed popularity in recent years because of their practical use in feedback control applications, such as fluid flow control and  sensor/actuator selection. 
The inherent difficulty of synthesizing SOF controllers is rooted in solving a series of non-convex problems that make the solution computationally intractable.
In this note, we show that SOF synthesis is a convex problem for the specific case of systems with a $lossless$ (i.e., energy-conserving) nonlinearity. 
Our proposed method ensures asymptotic stability of an SOF controller by enforcing the $lossless$ behavior of the nonlinearity using a quadratic constraint approach.
In particular, we formulate a bilinear matrix inequality~(BMI) using the approach, then show that the resulting BMI can be recast as a linear matrix inequality~(LMI).
The resulting LMI is a convex problem whose feasible solution, if one exists, yields an asymptotically stabilizing SOF controller.
\end{abstract}
\end{frontmatter}
\section{Introduction}
Static-output feedback (SOF) stabilization is an open controls problem \cite{bernstein1992, blondel2000survey, syrmos1997static}. 
The bilinearity of the SOF problem for linear time-invariant (LTI) systems makes the problem non-convex, 
rendering the solution intractable for high-dimensional LTI systems.
Despite the difficulties of SOF control synthesis, the simple structure of SOF control has made it an attractive choice in practice.
As such, various SOF synthesis algorithms have been developed for LTI systems~\cite{dinh2011combining,el1997cone,cao1998static,iwasaki1995xy,toivonen1985descent}.
%
On the contrary, SOF synthesis methods are less prevalent for nonlinear systems, wherein a similar bilinearity issue arises.
Nonetheless, SOF controllers would be desirable in many nonlinear applications.
%

\tm{Computing an SOF gain is NP-hard in general \cite{blondel1997}, 
i.e.,~the solution method for the SOF problem has a time complexity greater than $\mathcal{O}(n^k)$ for some positive $k$ and input size $n$;
therefore, exploiting the structure of a given system is the best approach to obtain a solution in a computationally tractable manner.}
No universally applicable algorithm exists for SOF synthesis of generic nonlinear systems; different systems possess different underlying structures that can be exploited.
%
For example, SOF synthesis methods have been developed for systems for which the nonlinearity is passive~\cite{Dalsmo1995}, polynomial~\cite{zhao2010robust,saat2015nonlinear} or Lipschitz~\cite{ekramian2020static}. 
%
Furthermore, sufficient conditions for local SOF stabilization have been established for nonlinear (affine) systems~\cite{astolfi2002hamilton}.

In this note, we will use standard Lyapunov arguments to establish a special class of system for which the SOF stabilization problem is convex: i.e., that of an LTI system in feedback with a \emph{lossless} nonlinearity or uncertainty element.
Such systems arise, for example, in the context of fluid dynamics, where it is well-established that the nonlinear terms in the incompressible Navier-Stokes equations are lossless and kinetic-energy conserving~\cite{Schmid2001}.
In fluid dynamics, the lossless (or more generally passivity) property of the nonlinearity has been exploited for stability analysis~\cite{Schmid2001, kalurAIAA2020,liu20,kalur2021nonlinear,kalur2021estimating,mushtaq2021feedback}, model reduction~\cite{antoulas2005approximation}, and feedback control synthesis~\cite{sharma2011,Heins2016}.
Further, the utility of (linear) SOF strategies for fluid flow control has also been investigated in recent works~\cite{yaoAIAAJ2022,yaoTCFD2022}

%


%
%
%
%

The work presented in this paper establishes convexity of SOF control synthesis for an LTI system interacting in feedback with a $lossless$ uncertainty element.
The $lossless$ property is a special case of passivity; energy is strictly conserved, neither being created nor dissipated.
We first use quadratic constraints~(QC) to bound the $lossless$ behavior of the nonlinearity. 
The QC allows us to formulate a bilinear matrix inequality~(BMI).
We then show that the particular BMI that arises in this problem has a structure that reduces it to a linear matrix inequality~(LMI).
The LMI condition establishes the convexity of the associated SOF synthesis problem.
In addition, we establish necessary and sufficient conditions for feasibility of our synthesis approach through standard application of the projection lemma.
%

The paper is organized as follows. Section 2 defines the QC approach, which is later used in Section 3 to establish the main result of this paper, i.e., explicit conditions for SOF synthesis. Section 4 contains an example problem and Section 5 is the conclusion.
\section{Formulating the Framework}
Consider a state-space system of the following form:
\begin{equation}
\begin{aligned}
    & \Dot{x}(t) = Ax(t) + Bu(t) + z(t) \\
    & y(t) = Cx(t) \\
    & z(t) = N(x(t))x(t),
\end{aligned}
\label{eq:dynamics}
\end{equation}
where the vectors $x(t) \in \R^n$, $u(t) \in \R^q$, and $y(t) \in \R^p$ are the states, control inputs, and system outputs, respectively. 
Moreover, $A \in \R^{n \times n}$, $B \in \R^{n \times q}$ and $C \in \R^{p \times n}$ are the system, control, and output matrices, respectively. 
\tm{Also, $N: \R^n \to \R^{n\times n}$ is a continuous function such that $N(0) = 0$ and each entry of $z(t) = N(x(t))x(t)$ is a continuous function in $x(t)$.
%
%
Here, we only consider \emph{lossless} nonlinearities for which $\left \langle x(t), z(t) \right \rangle = 0$ for all $x(t)$, \tm{where $\left \langle \cdot, \cdot \right \rangle$ is the inner-product of signals}.
Note that $\left \langle x(t), z(t) \right \rangle = x(t)^{\mathrm{T}} N(x(t)) x(t) = 0$ is a nonlinear constraint in $N$.
This constraint can be satisfied by any linear or nonlinear continuous function $N$ such that $N(x(t))$ is a skew-symmetric matrix for all $x(t)$.
However, other continuous functions are possible as long as $\left \langle x(t), z(t) \right \rangle = 0$ for all $x(t)$.}
Lossless nonlinearities are a special case of passive nonlinearities for which energy is conserved (i.e., energy is neither created nor dissipated)~\cite{khalil01}.
%
%
%
%
%

As shown in \cite{kalur2021nonlinear}, we can derive a matrix inequality condition using the Lyapunov stability theorem to capture the $lossless$ property.
Consider one such matrix inequality for the uncontrolled (open-loop) system studied in \cite{kalur2021nonlinear}, i.e., set $u(t) = 0$ in \eqref{eq:dynamics}:
 \begin{dmath}
\begin{bmatrix}
A^{\mathrm{T}} {P} + {P} A && {P} \\
{P} && {0} \\
\end{bmatrix}
+
\xi_o
\begin{bmatrix}
{0} && {I} \\
{I} && {0}
\end{bmatrix} \leq -\begin{bmatrix}
\epsilon {P} && {0} \\
{0} &&  {0}
\end{bmatrix}.
\label{eq:global_LMI_QC}
\end{dmath}
\tm{Here $\xi_o \in \mathbb{R}_{\lessgtr 0}$ is the Lagrange multiplier and $P\in \mathcal{S}_{++}^n$, where $\mathcal{S}_{++}^n$ is the set of $n\times n$ symmetric, positive definite matrices.
Feasibility of the matrix inequality \eqref{eq:global_LMI_QC} implies that $V(x(t))=x(t)^{\mathrm{T}} P x(t)$ is a Lyapunov function.
This can be verified by multiplying \eqref{eq:global_LMI_QC} on the left and right by $\begin{bmatrix}
    x(t)^{\mathrm{T}} && z(t)^{\mathrm{T}}
\end{bmatrix}$ and $\begin{bmatrix}
    x(t)^{\mathrm{T}} && z(t)^{\mathrm{T}}
\end{bmatrix}^{\mathrm{T}}$, respectively.
This gives the following constraint along trajectories of the system \eqref{eq:dynamics}:
\begin{align}
   \dot{V}(x(t)) +\xi_o x(t)^{\mathrm{T}} z(t) \le -\epsilon V(x(t))
\end{align}
Apply the lossless constraint to show that $\dot{V}(x(t)) \leq -\epsilon V(x(t))$, i.e., $V(x(t))$ decays with rate $\epsilon \in \R_{>0}$ along trajectories.}
%
%
%
%
Further details about the derivation are given in \cite{kalur2021nonlinear}.

The dynamics in \eqref{eq:dynamics} are open-loop stable for $u(t) = 0$ if there exists a $P > 0$ satisfying \eqref{eq:global_LMI_QC}.
Thus, by extension, \eqref{eq:dynamics} is closed-loop stable for $u(t) = Ky(t) = KCx(t)$ if there exists a $P > 0$ satisfying the following:
\begin{dmath}
\begin{bmatrix}
(A + BKC)^{\mathrm{T}} {P} + {P} (A+BKC) && {P} \\
{P} && {0} \\
\end{bmatrix}
+
\xi_o
\begin{bmatrix}
{0} && {I} \\
{I} && {0}
\end{bmatrix} \leq -\begin{bmatrix}
\epsilon {P} && {0} \\
{0} &&  {0}
\end{bmatrix},
\label{eq:global_LMI_QC_c}
\end{dmath}
where $K \in \R^{q \times p} $ is the SOF controller gain.
However, the solution is difficult to compute because the inequality in~\eqref{eq:global_LMI_QC_c} is bilinear in the variables $P$ and $K$.
To remedy the bilinearity issue, we show in section \ref{sec:3} that the BMI in \eqref{eq:global_LMI_QC_c} can be easily reduced to an equivalent LMI, which is convex and easily solvable.
\section{Output Feedback Controller (Main Result)}
\label{sec:3}
This section establishes conditions to solve \eqref{eq:global_LMI_QC_c} for an asymptotically stabilizing $K$.
Specifically, we show---using standard techniques---that \eqref{eq:global_LMI_QC_c} can be formulated as a convex problem in variable $K$ and that $P$ is a solution with a fixed structure.
Furthermore, we provide explicit feasibility conditions on 
the existence of a stabilizing $K$.

\begin{lemma}\label{lemma:1}
There exists $P \in \mathcal{S}^{n}_{++}$, $\xi_o \in \R_{\lessgtr 0}$, and $K \in \R^{q \times p}$ satisfying \eqref{eq:global_LMI_QC_c} if and only if there exists $K$ satisfying the following:
\begin{equation}
(A+BKC) + (A+BKC)^{\mathrm{T}} + \epsilon I \leq 0.   \label{eq:control_lmi}
\end{equation}
\end{lemma}
\begin{proof}
Notice that the lower right block of \eqref{eq:global_LMI_QC_c} is zero.
Thus, the off-diagonal terms must be zero in order for $\eqref{eq:global_LMI_QC_c}$ to be feasible.
Therefore, $P = -\xi_o I$ and $\xi_o < 0$.
Hence, \eqref{eq:global_LMI_QC_c} reduces to the following form:
\begin{equation}
(A + BKC)^{\mathrm{T}} (-\xi_o I) + ( -\xi_o I)(A + BKC) - \xi_o \epsilon I \leq 0.    
\label{eq:global_LMI_QC_sol2}
\end{equation}
%
%
Since $\xi_o$ is a negative scaling factor, we can eliminate it without affecting the solution to \eqref{eq:global_LMI_QC_sol2}.
%
%
Thus, the following LMI in $K$ results:
\begin{equation}
    (A + BKC)^{\mathrm{T}} + (A + BKC) + \epsilon I \leq 0.
    \label{eq:global_LMI_QC_sol3}
\end{equation}
%
%
Notice that any choice of $\xi_o < 0$ will satisfy the solution in~\eqref{eq:global_LMI_QC_c} since the resulting $P$ can always be re-scaled to an identity and the same argument used for deriving~\eqref{eq:global_LMI_QC_sol3} will follow.
Therefore, the solution to~\eqref{eq:global_LMI_QC_c} can be generalized to $P = I$ and $\xi_o = -1$.
The controller gain (if one exists) can be obtained by solving the feasibility problem in~\eqref{eq:global_LMI_QC_sol3}, which is now convex and is an equivalent condition to~\eqref{eq:global_LMI_QC_c}.
Therefore, the solution $K$ from \eqref{eq:global_LMI_QC_sol3} is a solution of \eqref{eq:global_LMI_QC_c} and vice versa.
\end{proof}
\begin{remark}
%
%
%
Note that we have not yet established the explicit conditions that guarantee the existence of $K$ as a feasible solution of \eqref{eq:control_lmi}.
\end{remark}

\begin{lemma}
\tm{There exists an $\epsilon \in \mathbb{R}_{>0}$ and an output-feedback gain $K \in \R^{q \times p}$ satisfying \eqref{eq:control_lmi} if and only if the following set of conditions are satisfied:}
\tm{\begin{equation}
    \begin{aligned}
        & U_B^{\mathrm{T}}(A^{\mathrm{T}} + A)U_B < 0\\
        & U_C^{\mathrm{T}}(A^{\mathrm{T}} + A)U_C < 0,
    \end{aligned}
    \label{eq:projection_lemma_conditions_v1}
\end{equation}}
\end{lemma}
\tm{where $U_{B} = null(B^{\mathrm{T}})$ and $U_{C} = null(C)$ are the null-spaces of $B^{\mathrm{T}}$ and $C$, respectively.}
\begin{proof}
%
\tm{We showed in Lemma \ref{lemma:1} that the following LMI results from the optimal solution $P = I$ and $\xi_o = -1$:
\begin{dmath}
    (A+BKC) + (A+BKC)^{\mathrm{T}} + \epsilon I \leq 0
    \label{eq:result_lmi_1}
\end{dmath}
which is equivalent to \eqref{eq:global_LMI_QC_c}.
Furthermore, the LMI in \eqref{eq:result_lmi_1} implies the following LMI for an arbitrarily chosen ${\epsilon \in \mathbb{R}_{>0}}$:
\begin{dmath}
    (A+BKC) + (A+BKC)^{\mathrm{T}} < 0.
    \label{eq:result_lmi_2}
\end{dmath}
Then, the LMI in \eqref{eq:result_lmi_2} can be written in an equivalent projection lemma form (see section 2.6.2 in \cite{Boyd1994}):
\begin{dmath}
A^{\mathrm{T}} + A + B K C + C^{\mathrm{T}} K^{\mathrm{T}} B^{\mathrm{T}}
< 0.
\label{eq:projection_lemma_form}
\end{dmath}
This is a particularly useful result as we can use the properties of the projection lemma to create two equivalent LMIs of the following form to solve for $K$:
\begin{equation}
    \begin{aligned}
        & U_B^{\mathrm{T}} \Psi U_B < 0 \\
        & U_C^{\mathrm{T}} \Psi U_C < 0
    \end{aligned}
    \label{eq:proj_form}
\end{equation}
where $\Psi = A^{\mathrm{T}} + A$, and $U_B$ and $U_C$ are the null-spaces of $
B^{\mathrm{T}}$ and $
C$, respectively. 
%
%
%
%
}
%
\end{proof}
\begin{remark}
\tm{If $B$ and $C$ are full-rank, then the associated null spaces $U_B$ and $U_C$ will only contain the trivial zero vector, respectively.}
In any of these cases, we can simply use a congruence transformation followed by a change of variables (\tm{i.e.,} $Y = PBK$ or $Y = KCP^{-1}$) to solve for the variables $P$ and $K$ \cite{Boyd2004}.
\end{remark}
\begin{remark}
In cases where the $lossless$ nonlinearity defines the global behavior of the system, the LMI \eqref{eq:control_lmi} gives a sufficient condition for global asymptotic stability.
\end{remark}
    %
    %
%
An SOF gain $K$ can be obtained simply by solving~\eqref{eq:control_lmi} using an LMI solver, such as \texttt{sdpt3} or \texttt{mosek}.
\section{Example Problem}
\tm{Consider the system in \eqref{eq:dynamics} with the following matrices for $x=\begin{bmatrix}
x_{1} && x_{2}\end{bmatrix}  
^{\mathrm{T}} \in \mathbb{R}^2:$}
\begin{equation}
\begin{aligned}
    & A =  \begin{bmatrix}
    -0.1 && 1 \\
    0 && -0.1
    \end{bmatrix}, 
    B = \begin{bmatrix}
    1 \\
    1 \\
    \end{bmatrix},
    \tm{N(x) = \begin{bmatrix}
    0 && -x_1 \\
    x_1 && 0
    \end{bmatrix}} \\
   & C = \begin{bmatrix}
    1 && 2
    \end{bmatrix}.
    \end{aligned}
    \label{eq:example}
\end{equation}
Then, $K \in \mathbb{R}$ is the SOF gain to compute.
We choose $\epsilon = 10^{-6}$ and obtain $K = -3.6231$ by solving \eqref{eq:global_LMI_QC_sol3} using the parser \texttt{cvx} and the LMI solver \texttt{sdpt3}~\cite{grant2008cvx}.
\begin{figure}[!htb]
    \centering
    \includegraphics[scale = 0.5]{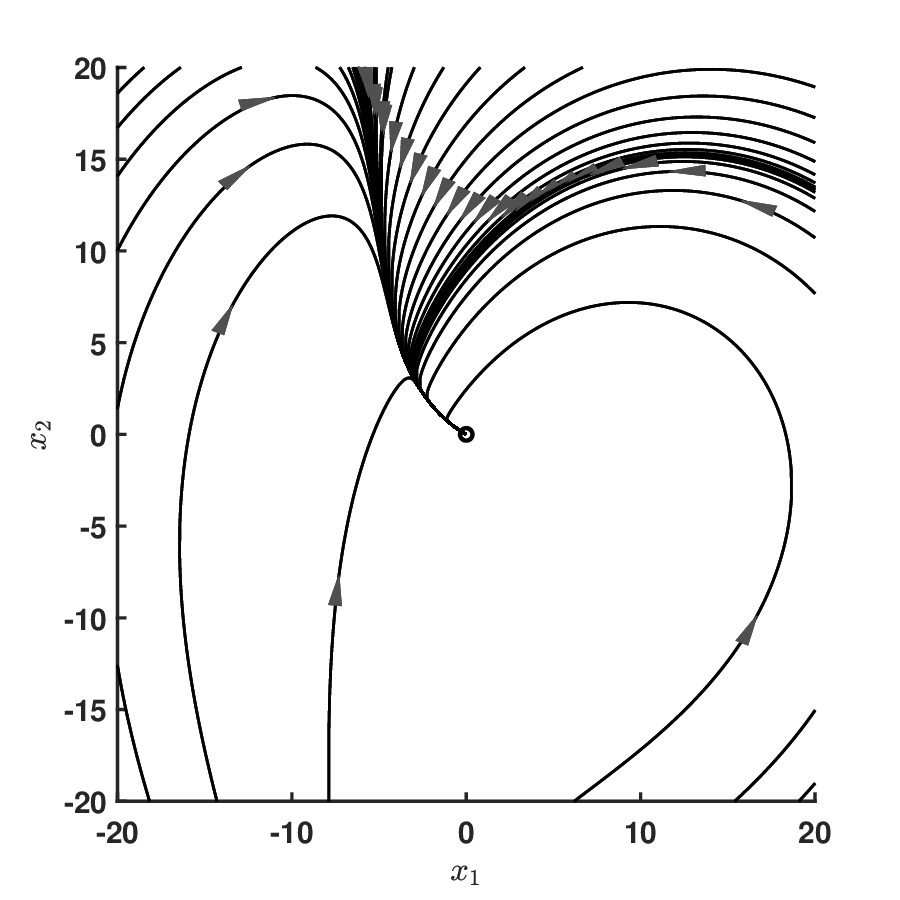}
    \caption{The phase portrait shows that the origin of the controlled system is asymptotically stable.}
    \label{fig:simulation}
\end{figure}
\begin{figure}[htbp]
    \centering
    \hspace*{-0.2cm}
    \includegraphics[scale = 0.5]{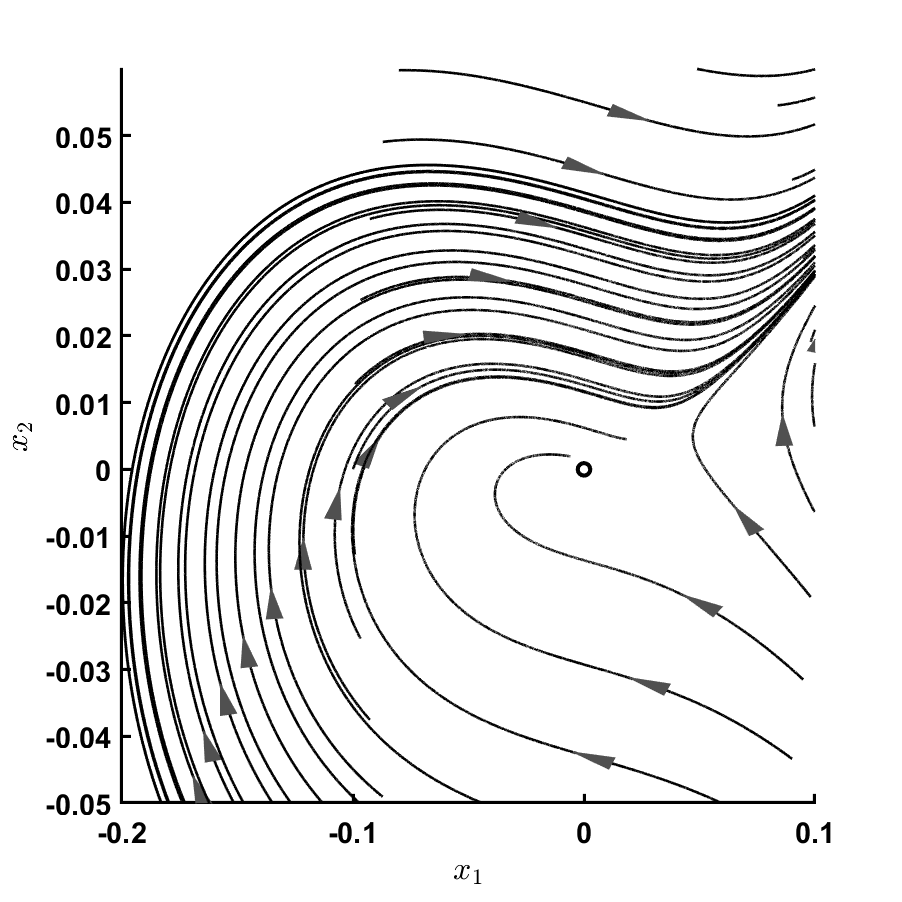}
    
    \caption{The phase portrait shows that the origin of the uncontrolled system is unstable, since some trajectories diverge.}
    \label{fig:simulation_2}
\end{figure}
Fig. \ref{fig:simulation} shows the phase portrait of the controlled (closed-loop) system.
%
%
We can see that $\|x(t)\|$  converges to the origin for all trajectories of the controlled (closed-loop) system.
%
%
%
On the contrary, Fig. \ref{fig:simulation_2} shows that some of the trajectories diverge from the origin for the uncontrolled (open-loop) system. 
Note the difference in scales between Fig. \ref{fig:simulation} and Fig. \ref{fig:simulation_2}: relatively small magnitudes $\|x(0)\|$ result in divergent trajectories for the uncontrolled (open-loop) system, whereas no divergent trajectories exist for the controlled system even for relatively large magnitudes $\|x(0)\|$.
%
This further confirms that the designed controller successfully provides asymptotic stability to the system in~\eqref{eq:example}.
%
%
%
%
%
%
%
\section{Conclusion}
We provide explicit conditions under which an SOF gain can be obtained using the $lossless$ behavior of the nonlinearity.
Thus, an SOF gain for any nonlinear system with a $lossless$ nonlinearity can be solved using the framework in this paper.
Note that the convexity of the problem was established as an LMI feasibility problem, which in general corresponds to a family of stabilizing SOF gains when the problem is feasible.
This non-uniqueness of stabilizing SOF gains opens additional avenues for optimal control using convex optimization methods.
Furthermore, the conditions established for feasibility of the LMI provide requirements on the set of sensors and actuators needed for synthesis.
In conjunction with convexity of the problem, these requirements can potentially be exploited to perform optimal sensor/actuator selection using established methods based on convex optimization~\cite{polyak2014sparse,Jovanovic2014ADMM, zare2018optimal}.
%

%
\section{Acknowledgements}
This material is based upon work supported by the ARO under grant number W911NF-20-1-0156 and the NSF under grant number CBET-1943988. 
%
MSH acknowledges support from the AFOSR under award number FA 9550-19-1-0034.
\bibliographystyle{IEEEtran}
\bibliography{ref}
\end{document}